\documentclass[11pt,twoside,letterpaper]{llncs}
\usepackage{times}
\usepackage[T1]{fontenc}   
\usepackage{amsmath}
\usepackage{amssymb}
\usepackage{fancybox}
\usepackage{multirow}
\usepackage{graphics,graphicx}
\usepackage{microtype}
\usepackage{booktabs}

\newcommand{\junk}[1]{}

\renewcommand{\baselinestretch}{1}

\begin{document}

\title{An output-sensitive algorithm for
all-pairs shortest paths in directed acyclic graphs}

\author{
Andrzej Lingas
\inst{1}
\and
Mia Persson
\inst{2}
\and
Dzmitry Sledneu
\inst{3}
\institute{
Department of Computer Science, Lund University, 22100 Lund, Sweden. 
\texttt{Andrzej.Lingas@cs.lth.se}
\and
Department of Computer Science and Media Technology, Malm\"o University, 20506 Malm\"o, Sweden.
\texttt{miapersson@mau.se}
\and
Malm\"o, Sweden.
\texttt{dzmitry.sledneu@gmail.com}}
}

\maketitle

\begin{abstract}
  A straightforward dynamic programming method for the single-source
  shortest paths problem (SSSP) in an edge-weighted directed acyclic
  graph (DAG) processes the vertices in a topologically sorted order.
  Yen, by decomposing the input edge-weighted directed graph in two
  DAGs, could use this method iteratively to improve the time
  complexity of the SSSP Bellman-Ford algorithm for edge-weighted
  directed graphs by a constant factor.

  First, we similarly iterate this method alternatively in a
  breadth-first search sorted order and the reverse order on an input
  directed graph with both positive and negative real edge weights,
  $n$ vertices and $m$ edges.
  For a positive integer $t,$
  after $O(t)$ iterations in $O(tm)$
  time, we obtain for each vertex $v$ a path distance from the source
  to $v$ not exceeding that yielded by the shortest path from the
  source to $v$ among the so called {\em$ t+$light paths}.  A directed
  path between two vertices is $t+$light if it contains at most $t$
  more edges than the minimum edge-cardinality directed path between
  these vertices.  After $O(n)$ iterations, we obtain an $O(nm)$-time
  solution to SSSP in directed graphs with real edge weights matching
  that of Bellman and Ford.

  Our main result is an output-sensitive algorithm for the all-pairs
  shortest paths problem (APSP) in DAGs with positive and negative
  real edge weights. It runs in time $O(\min \{n^{\omega},
  nm+n^2\log n\}+\sum_{v\in V}\text{indeg}(v)|\text{leaf}(T_v)|),$ where $n$ is
  the number of vertices, $m$ is the number of edges, $\omega$ is the
  exponent of fast matrix multiplication, $\text{indeg}(v)$ stands for
  the indegree of $v,$ $T_v$ is a tree of lexicographically-first
  shortest directed paths from all ancestors of $v$ to $v$, and
  $\text{leaf}(T_v)$ is the set of leaves in $T_v.$ Note that if $T_v$
  is a path the term $O(\text{indeg}(v)|\text{leaf}(T_v)|)$ equals
  $O(\text{indeg}(v))$ while if $T_v$ is a star with $v$ as a sink the
  term becomes $O(\text{indeg}(v)|T_v|).$ It also follows that if
  $\max_{v\in V} |\text{leaf}(T_v)|=O(n^{\alpha})$ then our APSP
  algorithm for DAGs runs in $O(n^{\omega}+mn^{\alpha})$ time.
  Similarly, if $\frac {\sum_{v\in V} |\text{leaf}(T_v)|}
  n=O(n^{\beta})$ then the algorithm runs in
  $O(n^{\omega}+n^{2+\beta})$ time.
  
Next, we discuss an extension of hypothetical improved upper
time-bounds for APSP in non-negatively edge-weighted DAGs to include
directed graphs with a polynomial number of large directed cycles.

Finally, we present experimental comparisons of our SSSP algorithm
with the Bellman-Ford one and our output-sensitive APSP algorithm for
edge-weighted DAGs with the standard APSP algorithm for edge-weighted
DAGs.  In particular, they show that our SSSP algorithm converges to
the true distances on dense edge-weighted pseudorandom graphs faster
than the Bellman-Ford algorithm does.
\end{abstract}
\vfill
\newpage
\section{Introduction}
The {\em length} of a path in an {\em edge-weighted} graph is the sum
of the weights of edges on the path. A {\em shortest} path between two
vertices in a graph has minimal length among all paths between these
vertices. The {\em distance} between vertices $v$ and $u$
is the length of a shortest path from $v$ to $u.$
If the graph is directed, the paths are supposed to be also
directed.

Shortest path problems, in particular the single-source shortest paths
problem (SSSP) and the all-pairs shortest paths problem (APSP), belong
to the most basic and important problems in graph algorithms
\cite{CLR,Z01}. There are several variants of SSSP and APSP depending
among other things on the restrictions on edge weights and the input
graphs.  The input to these problems is a directed or an undirected
edge-weighted graph. The output is a representation of shortest
paths between the source and all other vertices or between all pairs of
vertices in the graph, respectively.

In the general case of directed graphs (without negative cycles), when both positive and
negative real edge weights are allowed, the difference between the best
known asymptotic upper time-bounds for SSSP and APSP respectively is
surprisingly small. Namely, if the input directed graph has $n$
vertices and $m$ edges with real weights, then the best known SSSP
algorithm due to Bellman \cite{Bel58}, Ford \cite{For56}, and Moore
\cite{Mo59} runs in $O(nm)$ while the APSP can be solved already in
$O(nm+n^2\log n)$ time \cite{sur17,Z01}.  The APSP solution uses 
Johnson's $O(nm)$-time reduction of the general edge weight case to
the non-negative edge case and then it runs Dijkstra's algorithm
\cite{Dij59} $n$ times \cite{sur17,Z01}.  The latter upper time-bound
for APSP with arbitrary real edge-weights has been more recently improved to
$O(nm+n^2\log \log n)$ by Pettie in \cite{Pet04}.
Note that the aforementioned
best asymptotic upper time bounds for SSSP and APSP are different only for
sparse graphs with $o(n\log \log n)$ edges. Interestingly,
when edge weights are integers,
the best known upper time-bound for APSP 
just in terms of $n$ is  $n^3/2^{\Omega(\sqrt {\log n})}$
\cite{CW16}.

The situation alters dramatically when the input directed graph is
acyclic, i.e., when it does not contain directed cycles. Then, a
simple dynamic programming algorithm processing vertices in a
topologically sorted order solves the SSSP problem in $O(n+m)$ time
\cite{CLR}, an $O(n(n+m))$-time solution to the APSP problem in this case
follows.

In fact, Yen could use the aforementioned method for SSSP in DAGs
iteratively in order to improve the time complexity of Bellman-Ford
algorithm for directed graphs by a constant factor \cite{Yen70}.
Bellman-Ford algorithm runs in $n-1$ iterations. In each iteration, for
each edge $e$, the current distance (from the source) at the head
of $e$ is compared to the sum of the current distance at the tail of
$e$ and the weight if $e.$ If the sum is smaller the distance at the
head of $e$ is updated.  To achieve the improvement, Yen imposes a
linear order on the vertices of the input directed graph which yields
a decomposition of the graph into two DAGs.  Next, the SSSP method for
DAGs is run on each of the two DAGs instead of an iteration of
Bellman-Ford algorithm \cite{Yen70}.  Bannister and Eppstein
obtained a further improvement of the time complexity of Bellman-Ford
algorithm by a constant factor using a random linear order \cite{BE11}.

A pair of vertices in an edge weighted undirected or directed graph
can be connected by several paths, in particular several shortest
paths. Beside the length of a path, the number of edges forming it can
be an important characteristic.  For example, Zwick provided several
exact and approximation algorithms for all pairs {\em lightest} (i.e.,
having minimal number of edges) shortest paths in directed graphs with
restricted edge weights in \cite{Z99}.
\junk{Recall also the Bellman-Ford
algorithm for SSSP in directed graphs with positive and negative edge
weights. It runs in $n-1$ iterations, in each iteration for each edge
$e$ the current distance (from the source) at the head of $e$ is
compared to the sum of the current distance at the tail of $e$ and the
weight if $e.$ If the sum is smaller the distance at the head of $e$
is updated.
Note that after $k$ iterations the current distances are
not greater than the lengths of shortest corresponding paths using at
most $k$ edges.}

In this paper, first we consider {\em $t+$light paths}, i.e., directed
paths that have at most $t$ more edges than the paths with the same
endpoints having the minimal number of edges.  In part following
\cite{Yen70}, we iterate $O(t)$ times the SSSP method for DAGs on
two implicit
DAGs yielded by an extension of the BFS partial order to a
linear order.  The iterations alternatively process the vertices in a
breadth-first sorted order and the reverse order.  In result, we obtain
path distances from the source to all other vertices that are not
greater than the corresponding shortest-path distances for $t+$light
paths.  It takes $O(tm)$ time totally. For $t=n-2$, our method matches
that of Bellman-Ford for SSSP in directed graphs with real edge
weights.

A vertex $v$ is an {\em ancestor} ({\em direct ancestor},
respectively) of a vertex $u$ in a DAG if there is a directed path
(edge, respectively) from $v$ to $u$ in the DAG.

Our main result is an output-sensitive algorithm for the
APSP problem in DAGs.
It runs in time
  $O(\min \{n^{\omega}, nm+n^2\log n\}+\sum_{v\in V}\text{indeg}(v)|\text{leaf}(T_v)|),$ where
$n$ is the number of vertices, $m$ is the number of edges,
$\omega$ is the exponent
of fast $n\times n$ matrix multiplication
\footnote{$\omega$ is not greater than $2.3729$ 
\cite{AV21}.},
  $\text{indeg}(v)$ stands for the indegree
  of $v,$ $T_v$ is a tree of lexicographically-first shortest
  directed paths from all ancestors of $v$ to $v$, $\text{leaf}(T_v)$ is
  the set of leaves in $T_v,$ and for a set $X$, $|X|$ stands for its size.
    Note that if $T_v$ is a path
  the term $O(\text{indeg}(v)|\text{leaf}(T_v)|)$ equals
  $O(\text{indeg}(v))$ while when $T_v$ is a star with $v$ as a
  sink the term becomes $O(\text{indeg}(v)|T_v|).$
Thus, the running time of the APSP algorithm
can be so low as $O(n^{\omega})$ and so high as $O(n^{\omega}+nm).$
It follows also that if $\alpha$ is defined by
$\max_{v\in V} |\text{leaf}(T_v)|=O(n^{\alpha})$
  then the algorithm
  runs in $O(n^{\omega}+mn^{\alpha})$ time.
  Similarly, if $\beta $ is defined by $\frac {\sum_{v\in V} |\text{leaf}(T_v)|} n=O(n^{\beta})$
  then the algorithm runs in $O(n^{\omega}+n^{2+\beta})$ time.

Next, we provide an extension of hypothetical, improved upper
time-bounds for APSP in DAGs with non-negative edge weights to include
directed graphs with a polynomial number of large directed cycles.

Finally, we present experimental comparisons of our SSSP
algorithm with the Bellman-Ford one and our output-sensitive APSP algorithm
for edge-weighted DAGs with the standard APSP algorithm for edge-weighted DAGS.
In particular, they show that our SSSP algorithm
converges to the true shortest-path distances
on dense edge-weighted pseudorandom graphs faster than
the Bellman-Ford algorithm does. On the other hand, they 
exhibit only a slight time-performance advantage of our APSP algorithm 
over the standard APSP algorithm on dense edge-weighted pseudorandom DAGs.
Presumably, the shortest-path trees in the aforementioned DAGs
have large number of leaves.
\subsection{Paper organization}
In the next section, we provide our solution to the SSSP problem in
directed graphs with real edge weights based on the SSSP method for
DAGs and the BFS partial order in terms of $t+$light paths.
Section 3 is devoted
to our output-sensitive algorithm for the APSP problem in DAGs with
real edge weights and its analysis. In Section 4, we discuss the
extension of hypothetical, improved bounds for APSP in DAGs with
non-negatively weighted edges to directed graphs with a
polynomial number of large directed cycles.
Section 5 presents our experimental results.
We conclude with final
remarks.

\section{An application of the SSSP method for DAGs}

The SSSP problem for directed acyclic graphs can be solved by
topologically sorting the DAG vertices and applying straightforward  dynamic
programming.  For consecutive vertices $v$ in the sorted order, the distance
$dist(v)$ of $v$ from the source
is set to the minimum of $dist(u)+weight(u,v)$ over all
direct ancestors $u$ of $v$, where $weight(u,v)$
stands for the weight of the edge $(u,v)$.  It takes linear (in the size of the DAG)
time.  Yen used the dynamic programming method iteratively to improve
the time complexity of Bellman-Ford algorithm for directed graphs by a
constant factor in \cite{Yen70}.  Interestingly, we can similarly
apply this method iteratively to determine shortest-path distances
among paths using almost the minimal number of edges. To formulate our
algorithm (Algorithm 1), we need the following definition and two
procedures.

\begin{definition}
A directed path from a vertex $u$ to a vertex $v$
in a directed graph is {\em lightest} if it consists
of the smallest possible number of edges.
A path from $u$ to $v$ is $t+$light if
it includes at most $t$ more edges than
a lightest path from $u$ to $v.$
\end{definition}
\par
\noindent
{\bf procedure} $SSSPDAG(G,D)$
\par
\noindent
{\em Input:} A directed graph $(V,E)$ 
with real edge weights, linearly ordered
vertices $v_1,....,v_n,$
and
a $1$-dimensional table $D$ of size $n$
with upper bounds on the distances
from $v_1$ to all vertices in $V.$
\par
\noindent
{\em Output:} Improved upper bounds on the shortest-path distances from
$v_1$ to all vertices in $V$ in the table $D.$
\par
\vskip 2pt
\noindent
{\bf for} $j=2,...,n$ {\bf do}
\par
\noindent
For each edge $(v_i,v_j)$ where $i<j$\\
$D(v_j) \leftarrow \min\{ D(v_j),D(v_i)+weight(v_i,v_j)\}$
\par
\vskip 5pt
\noindent
{\bf procedure} $reverseSSSPDAG(G,D)$
\par
\vskip 3pt
\noindent
{\em Input and output:} the same as in $SSSPDAG(G,D)$
\par
\vskip 2pt
\noindent
{\bf for} $j=n-1,...,1$ {\bf do}
  \par
  \noindent
For each edge $(v_i,v_j)$ where $i>j$\\
$D(v_j) \leftarrow \min\{ D(v_j),D(v_i)+weight(v_i,v_j)\}$
\par
\vskip 6pt
\noindent
{\bf Algorithm 1}
\par
\noindent
{\em Input:} A directed graph $(V,E)$ 
with $n$ vertices, real edge weights and a distinguished
source vertex $s$, and a positive integer $t.$
\par
\noindent
{\em Output:} Upper bounds on the shortest-path distances from
$s$ to all other vertices in $V$ not exceeding the
corresponding shortest-path
distances constrained to  $t+$light paths.
\begin{enumerate}
\item Run BFS from the source $s$.
\item Order the vertices of $G$
  extending the BFS partial order according to the levels
  of the tree, i.e., $s$ comes first,
  then the vertices reachable by direct edges from $s$,
  then the vertices reachable by paths composed of two
  edges  and so on. We may assume w.l.o.g. that all vertices
  are reachable from $s$ or alternatively extend the aforementioned
  order with the non-reachable vertices arbitrarily.
\item Initialize a $1$-dimensional
  table $D$ of size $n,$
  setting $D(v_1)\leftarrow 0$ and $D(v_j)\leftarrow\infty$ for $1<j\le n$
\item $SSSPDAG(G,D)$
\item {\bf for} $k=1,...,t $ {\bf do}
  \begin{enumerate}
\item $reverseSSSPDAG(G,D)$
\item $SSSPDAG(G,D)$
  \end{enumerate}
\end{enumerate}

\junk{
\begin{figure}
\label{fig: dag1}
\begin{center}
\includegraphics[scale=0.5]{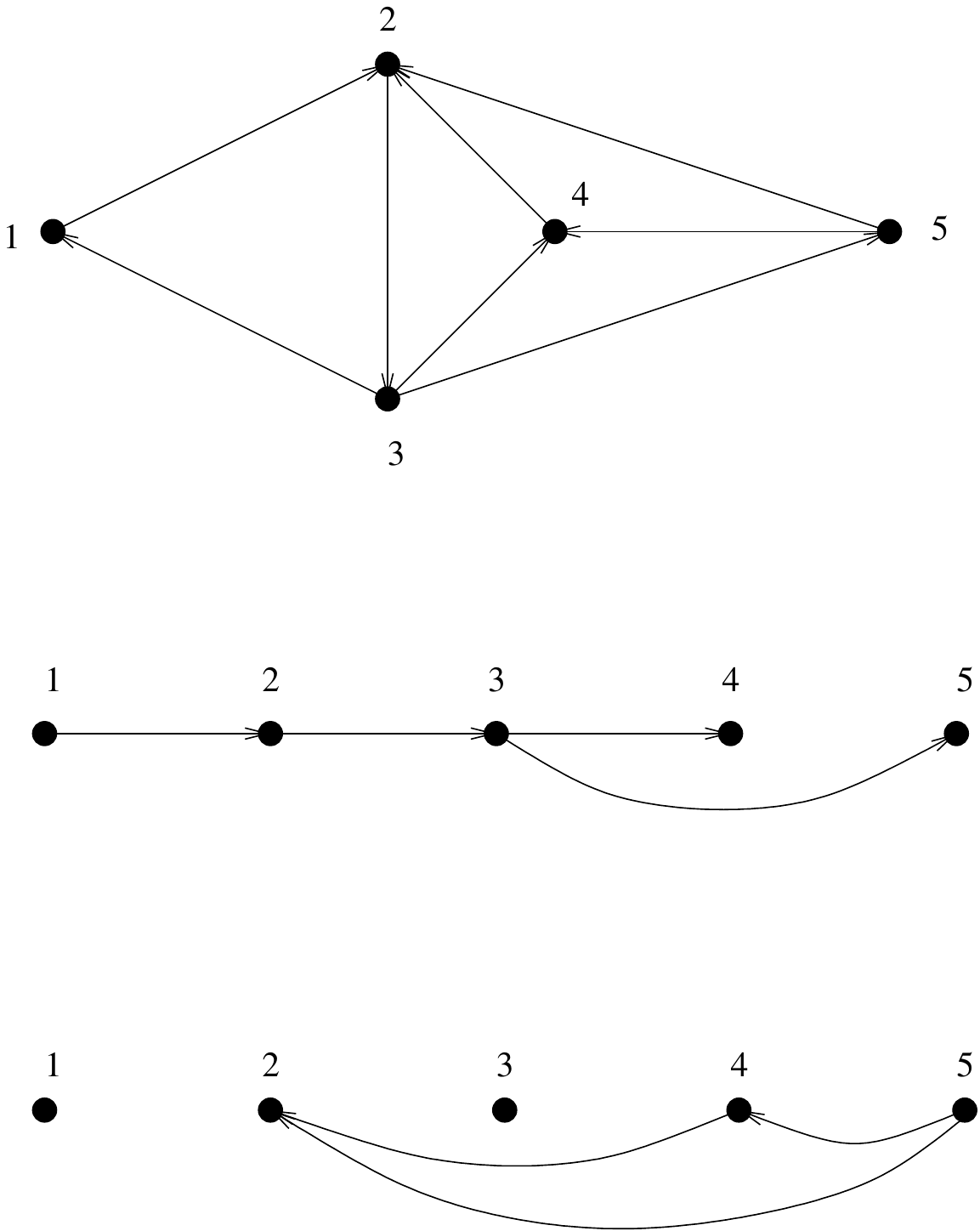}
\end{center}
\caption{An example of a graph with a BFS vertex numbering and the
two DAGs implied by forward and backward edges, respectively.}
\end{figure}}

\begin{theorem}\label{theo: light}
Let $G$ be a directed graph with $n$ vertices,
$m$ real-weighted edges, and a distinguished source vertex $s.$
For all vertices $v$ of  $G$ different from $s$,
an upper bound on their distance
from the source vertex $s$, not exceeding the length of a shortest path
among $t+$light paths from $s$ to $v,$
can be computed in $O((t+1)(m+n))$ total time.
\end{theorem}
\begin{proof}
  Consider Algorithm 1 and in particular
  the ordering of the vertices specified in its second step.
  We shall refer to an edge $(v_i,v_j)$ as forward if
  $i<j$ otherwise we shall call it backward.
  Note that the vertices at the same level of the BFS tree
  can be connected both by forward as well as backward edges.
  See also Fig. 1.
  Let $\ell$ be the number of (forward) edges
  in a lightest path from $s$
  to a given vertex $v.$
  It follows that any path from $s$ to $v$,
  in particular a shortest $t+$light one,
  has to have at least $\ell$ forward edges.

  Consider the BFS tree from the source $s.$
  Define the level of a vertex in the tree as
  the number of edges on the path from $s$
  to the vertex in the tree. Thus, in particular,
  $level(s)=0$ while $level(v)=\ell.$
  Recall that the linear order
  extending the partial BFS order used in
  Algorithm 1 is non-decreasing with respect of
  the levels of vertices. Also, if $(u,w)$
  is a forward edge then $level(u)\le level(w)\le level(u)+1$
  and if $(u,w)$ is a backward edge then $level(u)\ge level(w).$
  Hence, any path from $s$ to $v$ has to have at least
  $\ell$ forward edges, each increasing the level by one.
  
Consequently, a shortest  $t+$light path from $s$ to $v$
can have at most $t$ backward edges.
Thus, it can be decomposed into at most $2t+1$ maximal
fragments of consecutive edges of the same type (i.e.,
forward or backward, respectively), where the even
numbered fragments consist of backward edges.
Thus, the at most $2t+1$
calls of the procedures
$SSSPDAG(G,s,D)$,\\
$reverseSSSPDAG(G,s,D)$ in the algorithm
are sufficient to detect a distance from $s$ to $v$
not exceeding the length of a shortest path 
among $t+$light paths from $s$ to $v.$
The asymptotic running time of the algorithm is dominated
by the aforementioned procedure calls.
Hence, it is $O((t+1)(m+n))$.
\end{proof}

\begin{figure}
\label{fig: dag1}
\begin{center}
\includegraphics[scale=0.5]{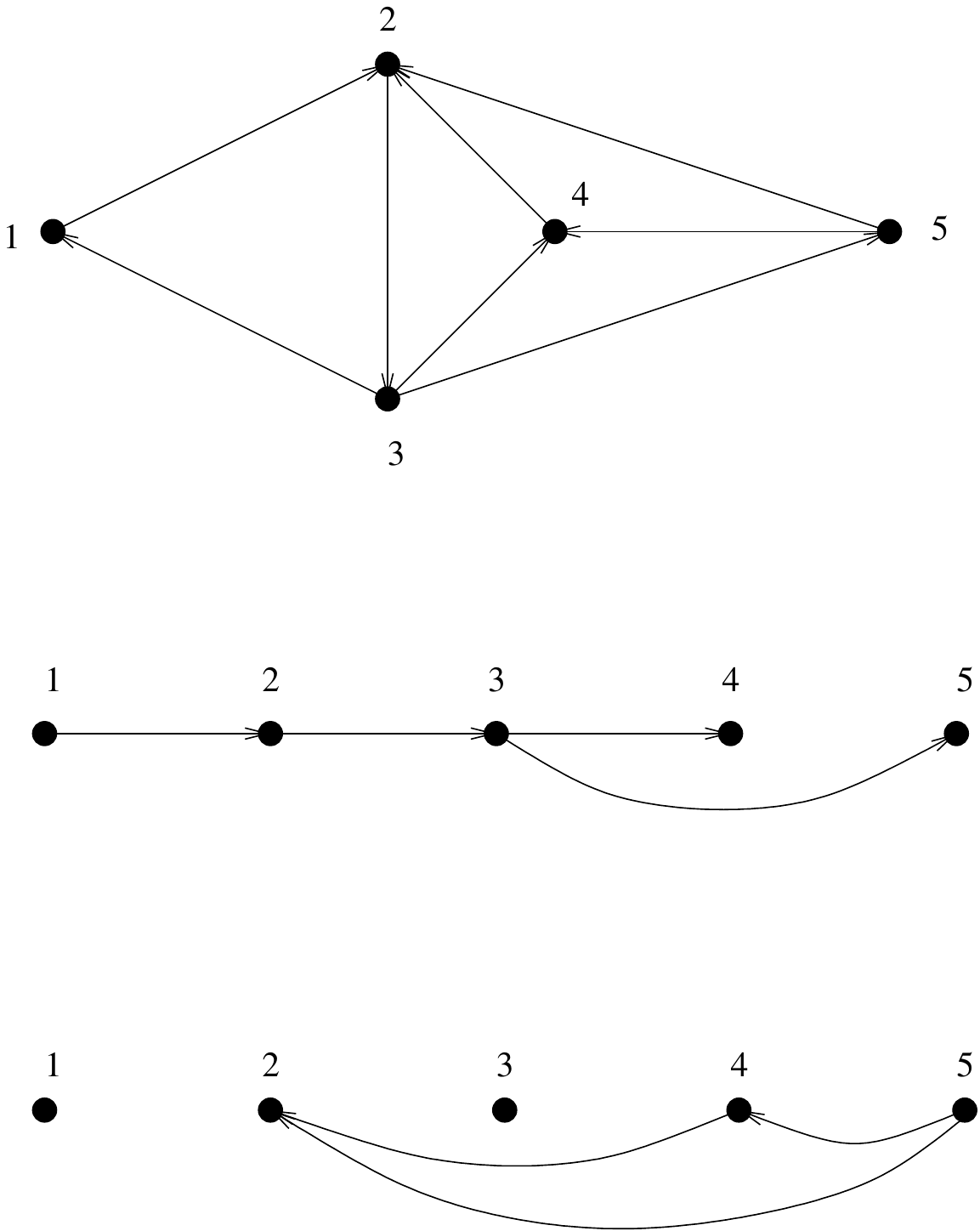}
\end{center}
\caption{An example of a graph with a BFS vertex numbering and the
two DAGs implied by forward and backward edges, respectively.}
\end{figure}

We can obtain a representation of directed paths achieving the upper
bounds on the distances from the source provided in Theorem \ref{theo:
  light} in a form of a tree of paths emanating from the source by
backtracking.  By setting $t=n-2$ in this theorem, we can match the
best known SSSP algorithm for directed graphs with positive and
negative real edge weights, i.e., the Bellman-Ford algorithm and
its constant factor improvements \cite{sur17,Z01}, running in $O(nm)$
time.  Similarly as in the case of Bellman-Ford algorithm, by calling
additionally $reverseSSSPDAG(G,D)$ and $SSSPDAG(G,D)$ after the
last iteration in Algorithm 1, we can detect the existence of negative
cycles.

Comparing our algorithm with the Bellman-Ford one, note that if the
lightest path from the source to a vertex $v$ has $\ell$ edges then
$\ell +t$ iterations in the Bellman-Ford algorithm may be needed to
obtain an upper bound on the distance of $v$ from the source comparable to
that obtained after $O(t)$ iterations in Algorithm 1.

\junk{
\begin{theorem}
The SSSP problem for directed graphs with $n$ vertices,
$m$ edges and non-negative edge weights
of $O(1)$ value can be solved in $O(nm)$ time.
\end{theorem}
\begin{proof}
Use Dijkstra's' algorithm with a $1$-dimensional table of size $O(n)$
instead of a priority queue.
\end{proof}}
\section{An output-sensitive  APSP algorithm for DAGs}

The APSP problem in DAGs with both positive and negative real edge
weights can be solved in $O(n(n+m))$ time by running $n$ times the SSSP
algorithm for DAGs. It is an intriguing open problem if there exist
substantially more efficient algorithms for APSP in edge-weighted
DAGs. In this section, we make a progress on this question by
providing an output-sensitive algorithm for this problem.
Its running time depends on the structure of shortest path
trees. Although in the worst-case it does not break the $O(nm)$
barrier it seems to be substantially more efficient in the majority of
cases.

The standard algorithm for APSP for DAGs just runs the SSSP
algorithm for DAGs for
each vertex of the DAG as a source separately. Our APSP algorithm  does
everything in one sweep along the topologically sorted order. Its main
idea is for each vertex $v$ to compute the tree of lexicographically-first
shortest paths from the ancestors $u$ of the currently processed
vertex $v$ to $v$, in the topologically sorted order. In case the tree
of lexicographically-first shortest paths from the already considered
ancestors of $v$ includes $u$ (as some intermediate vertex) then we
are done as for $u.$ Otherwise, we have to find the direct ancestor of
$v$ on the lexicographically-first shortest path $P$ from $u$ to $v$
and add an initial fragment of $P$ to the tree. By the topologically
sorted order in which the ancestors $u$ of $v$ are considered, this
can happen only when $u$ is a leaf of the (final) tree of
lexicographically-first shortest paths from the ancestors of $v$ to $v.$ The
direct ancestor of $v$ on $P$ can be found by comparing the lengths of shortest
paths from $u$ to $v$ with different direct ancestors of $v$ as the
next to the last vertex on the paths in time proportional to the
indegree of $v.$ In turn, the initial fragment of $P$
to add can be found by using the
link to the lexicographically-first shortest path from $u$ to the direct
ancestor of $v$ that is on $P.$ The correctness of the algorithm is
immediate.  The issues are an implementation of these steps and an
estimation of the running time.

To specify our
output-sensitive algorithm (Algorithm 2) more
exactly, we need the following definition.

\begin{definition}
  Assume a numbering of vertices in an edge-weighted DAG
  extending the topological partial order.
  A shortest (directed) path $P$ from $v_k$ to $v_i$ in the DAG
  is {\em first in a lexicographic order} if the direct
  ancestor $v_j$ of $v_i$ on $P$ has the lowest number $j$
  among all direct ancestors of $v_i$ on shortest
  paths from $v_k$ to $v_i$  and the subpath of $P$
  from $v_k$ to $v_j$ is  the lexicographically-first
  shortest path from $v_k$ to $v_j.$ For a vertex $v_i$ in the DAG,
  the tree $T_{v_i}$ of (lexicographically-first) shortest paths
  is the union of lexicographically-first paths from
  all ancestors of $v_i$ to $v_i.$ Note that the vertex $v_i$ is a sink
  of $T_{v_i}.$ It is assumed to be the root of $T_{v_i}$ and
  $\text{leaf}(T_{v_i})$ stands for the set of leaves of $T_{v_i}.$
 \end{definition}
 \par
\vskip 4pt
\noindent
{\bf Algorithm 2}
\par
\noindent
{\em Input:} A DAG  $(V,E)$ 
with real edge weights.

\par
\noindent
{\em Output:} For each vertex $v\in V,$
the tree $T_v$ of lexicographically-first shortest paths from all 
ancestors of $v$ to $v$ given by the table $NEXT_v$, where for
each ancestor $u$ of $v$, $NEXT_v(u)$ is the direct
successor of $u$ in the tree $T_v,$ (i.e., the head of the unique
directed edge having $u$ as the tail in the tree).
\begin{enumerate}
\item
  Determine the source vertices,
topologically sort the remaining vertices in $V$,
and number the vertices in $V$ accordingly,
assigning to the sources the lowest numbers.
\item Set $n$ to $|V|$ and $r$ to the number
  of sources in $G.$
\item Initialize an $n\times n$ table
  $dist$ by setting $dist(u,u)=0$ 
  and $dist(u,v)=\infty$ for $u,\ v \in V, \ u\neq v.$
\item {\bf for} $i=r+1,...,n$ {\bf do}
\begin{enumerate}
\item Compute
  the set $A(v_i)$ of ancestors of $v_i.$
\item Initialize a $1$-dimensional table $NEXT_{v_i}$ of size $|A(v_i)|$,\\
  setting $NEXT{v_i}(v_j)$ to $0$ for $v_j\in A(v_i)$.
\item {\bf for} $v_k\in A(v_i)$ in increasing order of the index $k$ {\bf do}
\begin{enumerate}
\item {\bf if} $NEXT_{v_i}(v_k)\neq 0$ {\bf then}
proceed to the next iteration of the interior for block.
\item Determine a direct ancestor $v_j$ of $v_i$ that
minimizes the value of $dist(v_k, v_j)+ weight(v_j,v_i)$.
In case of ties the vertex $v_j$ with the smallest index $j$ is
chosen among those yielding the minimum.
\item
$v_{current}\leftarrow v_k$
\item
  {\bf while} $v_{current}\neq v_j$ $\&$  $NEXT(v_{current}, v_i)=0$ {\bf do}\\
  $dist(v_{current}, v_i)\leftarrow dist(v_{current}, v_j) + weight(v_j, v_i$)\\
 $NEXT_{v_i}(v_{current})\leftarrow NEXT_{v_j}(v_{current})$\\
  $v_{current}\leftarrow NEXT_{v_i}(v_{current})$
\item
  {\bf if}  $NEXT_{v_i}(v_j)=0$ {\bf then} $dist(v_j, v_i)\leftarrow weigh(v_j, v_i)$ $\&$
  $NEXT_{v_i}(v_j)\leftarrow v_i$
\end{enumerate}
\end{enumerate}
\end{enumerate}

\begin{lemma}\label{lem: path}
  Steps 4.c.iii-v add the missing fragments of a lexicographically
  shortest path from $v_k$ to $v_i$ and set
  the distances from vertices in the fragments to $v_i$ in time proportional
to the number of vertices added to $T_{v_i}$.
\end{lemma}
\begin{proof}
Follow the path from $v_k$ to $v_j$ in $T_{v_j}$
extended by $(v_j,v_i)$ until
a vertex $v_q\in T_{v_i}$ is encountered.
This is done in Steps 4.c.iii-v.
The membership of $v_{current}$ in $T_{v_i}$
is verified by checking whether or not
$NEXT_{v_i}(v_{current})=0.$
Also, if $v_{current}$ is not yet in $T_{v_i}$
then its distance to $v_i$ is set
by  $dist(v_{current}, v_i)\leftarrow dist(v_{current}, v_j) + weight(v_j, v_i)$
and it is added to $T_{v_i}$ by
$NEXT_{v_i}(v_{current})\leftarrow NEXT_{v_j}(v_{current})$
in Step 4.c.iv.
By the inclusion of $v_q$ in $T_{v_i}$,
a whole shortest path $Q$ from $v_q$ to $v_i$
is already included in $T_{v_i}$ by induction
on the number of steps performed by the algorithm.
We claim that $Q$ exactly overlaps with
the final fragment of the extended path starting from $v_q$.
To see this encode $Q$ and the aforementioned
fragment of the extended path
by the indices of their vertices in the reverse
order. By our rule of resolving ties in Step 4.c.ii
both encodings should be first in the
lexicographic order so we have an exact overlap.
For this reason, it is sufficient to add
the initial fragment of the extended path
ending at $v_q$ to $T_{v_i}$ and if necessary
also the edge $(v_j,v_i)$ to $T_{v_i},$ and to update
the distances from vertices in the added fragment to $v_i,$
i.e., to perform Steps 4.c.iii-v.
\end{proof}

\begin{theorem} \label{theo: main}
The APSP algorithm for a
DAG $(V,E)$ with $n$ vertices, $m$ edges and real edge  weights
(Algorithm 2) runs in time
$O(\min \{n^{\omega}, nm+n^2\log n\}+\sum_{v\in V}\text{indeg}(v)|\text{leaf}(T_v)|).$
\end{theorem}
\begin{proof}
  The sets of ancestors can be determined in Step 4.a
  by computing the transitive closure of the input
  DAG in $O(\min \{n^{\omega}, nm\})$ time by using fast matrix multiplication \cite{Mu71}
  or BFS \cite{CLR}, first.
  In fact, to implement the loop in Step4.c,
  we need the sets of ancestors to be ordered according to the numbering
  of vertices provided in Step 1. If the transitive closure matrix
  is computed such an ordered set of ancestors  can be easily
  retrieved in $O(n)$ time. Otherwise, additional
  preprocessing sorting the unordered sets of ancestors is needed.
  The total cost of the additional preprocessing is $O(n^2\log n).$
  
  All the remaining steps, excluding Steps 4.c.ii-v for vertices
  $v_k$ not yet in $T_{v_i}$,
  can be done 
  in total (i.e., over all iterations) time
  $O(\sum_{v\in V}(1+|A(v)|))=O(n^2),$ where $A(v)$
  stands for the set of ancestors of $v$ in the DAG.
  The time taken by Step 4.c.ii,
when $v_k$ is not yet in the current $T_{v_i}$, is $O(\text{indeg}(v_i)).$
Suppose that $v_k$ is not a leaf of the final tree $T_{v_i}$.
Then, there must exist some leaf $v_p$ of the final tree
such that there is
path from $v_p$ via $v_k$ to $v_i$ in this tree.
By the numbering of vertices extending the partial topological
order, we have $p< k.$ We infer that the aforementioned
path is already present in the current $T_{v_i}$.
Thus, in particular the vertex $v_k$ is in the current tree.
Hence, the total time taken
by Step 4.c.ii is $O(\sum_{v\in V}\text{indeg}(v)|\text{leaf}(T_{v})|).$
Finally,
the total time taken by Steps 4.c.iii-v is $O(\sum_{v\in V}(1+|A(v)|))$
by Lemma \ref{lem: path}.
\end{proof}

Note that
the following inequalities hold:
$$\sum_{v\in V}\text{indeg}(v)|\text{leaf}(T_v)|\le m\max_{v\in V} |\text{leaf}(T_v)|,$$
$$\sum_{v\in V}\text{indeg}(v)|\text{leaf}(T_v)|\le n^2\frac {\sum_{v\in V} |\text{leaf}(T_v)|} n.$$
They immediately yield the following corollary from Theorem \ref{theo: main}.

\begin{corollary}
  Let $G=(V,E)$ be an $n$-vertex
  DAG with $n$ vertices and $m$ edges
  with real edge  weights.
  Suppose $\max_{v\in V} |\text{leaf}(T_v)|=O(n^{\alpha})$
  and\\
  $\frac {\sum_{v\in V} |\text{leaf}(T_v)|} n=O(n^{\beta})$.
The APSP problem for $G$  is solved by
Algorithm 2 in time
  $O(\min \{n^{\omega}, nm+n^2\log n\}+\min \{ mn^{\alpha},n^{2+\beta}\}).$
\end{corollary}

Observe that $ |\text{leaf}(T_v)|$ is equal to the minimum
number of directed paths covering the tree $T_v.$
Hence, $\alpha < 1$ if the maximum of the minimum number
of paths covering $T_v$ over $v$ is substantially sublinear.
Similarly, $\beta < 1$ if the average of the minimum number
of paths covering $T_v$ over $v$ is substantially sublinear.

To illustrate the superiority of Algorithm 2 over
the standard $O(n(n+m))$-time method for APSP in DAGs, consider
the following simple, extreme example.
\par
\noindent
Suppose $M$ is a positive integer.
Let $D$ be a DAG with vertices $v_1, v_2$,...,$v_n,$
and edges $(v_i,v_j),$ where $i<j,$ such that
the weight of $(v_i,v_j)$ is $-1$ if $j=i+1$ and $M$ otherwise.
\par
\noindent
It is easy to see the tree $T_{v_i}$ is just the path
$v_1,\ v_2,...,v_i$ and hence $|\text{leaf}(T_{v_i})|=1.$
Consequently, Algorithm 2 on the DAG $D$ runs in $O(n^{\omega})$
time while the standard method requires $O(n^3)$ time.
If $M=1,$
one could also run Zwick's APSP algorithm for directed
graphs with edge weights in $\{-1,0,1\}$ on this example
in $O(n^{2.575})$ time \cite{Z02}.

To refine Theorem \ref{theo: main}, we need the following
definition.

\begin{definition}
  For an edge weighted DAG $G$, let
  $\tilde{G}$ be the edge weighted DAG resulting from
  reversing the direction of edges in $G.$
  For a vertex $v_i$ in the DAG $G$,
  the tree $U_{v_i}$ of (lexicographically-first in reversed order) shortest paths from $v_i$
  to all descendants of $v_i$ in $G$ is just the tree
  resulting from the tree $T_{v_i}$ in $\tilde{G}$
  by reversing the edge directions.
  Note that the vertex $v_i$ is a source
  of $U_{v_i}.$ It is assumed to be the root of $U_{v_i}$ and
  $\text{leaf}(U_{v_i})$ stands for the set of leaves of $U_{v_i}.$
\end{definition}

The APSP for edge weighted DAGs can be solved
by providing the trees $U_v$ instead of the trees $T_v.$
Hence, we obtain immediately the following strengthening
of Theorem \ref{theo: main} by symmetry.

\begin{theorem} \label{theo: gmain}
The APSP problem for a
DAG $(V,E)$ with $n$ vertices, $m$ edges and real edge  weights
can be solved in time
$$O(\min \{n^{\omega}, nm+n^2\log n\}+
\min \{ \sum_{v\in V}\text{indeg}(v)|\text{leaf}(T_v),
\sum_{v\in V}\text{outdeg}(v)|\text{leaf}(U_v)\}
).$$
\end{theorem}
\begin{proof}
  Alternate the steps of Algorithm 2 run on the input DAG $G$
  with those of Algorithm 2 run on the DAG $\tilde{G}$.
  When any of the two runs finishes we are basically
  done. In case the run of Algorithm 2 on the DAG $\tilde{G}$
  finishes first, we obtain the trees $U_v$ in $G$
  from the trees $T_v$ in $\tilde{G}$ by
  reversing the direction of edges.
  The upper time bound follows from Theorem \ref{theo: main}
  and the fact that the indegree of a vertex in 
  $\tilde{G}$ is equal to its outdegree in $G.$
  \end{proof}
\section{A potential extension to digraphs with large cycles}
As we have already
noted the APSP problem in DAGs with both positive and negative real edge
weights can be solved in $O(n(n+m))$ time.
It is also an interesting open problem
if one can derive substantially  more efficient algorithms for APSP in DAGs
than the $O(n(n+m))$-time method in case of
restricted edge weights, e.g., non-negative edge weights etc.
In this section, under the assumption of the existence
of such substantially more efficient algorithms for DAGs
with non-negative edge weights, we show that 
they could be
extended to include directed graphs having
a polynomial number of large cycles.

The idea of the extension is fairly simple, see Fig. 2.  We pick
uniformly at random a sample of vertices of the input directed graph
that hits all the directed cycles with high probability
(cf. \cite{Z02}). Here, we use the assumption on the minimum size of
the cycles and on the polynomially bounded number of the
cycles. Next, we remove the vertices belonging to the sample and run
the hypothetical fast algorithm for APSP in DAGs on the resulting
subgraph of the input graph which is acyclic with high probability.
In order to take into account shortest path connections using the
removed vertices, we run the Dijkstra's SSSP algorithm from each
vertex in the sample on the original input graph two times. In the
second run we reverse the directions of the edges in the input
graph. Finally, we update the shortest path distances
appropriately. See Algorithm 3 for a more detailed description.
\par
\vskip 4pt
\noindent
{\bf Algorithm 3}
\par
\noindent
    {\em Input:} A directed graph $(V,E)$ with $n$ vertices,
    $m$ non-negatively weighted edges
and a polynomial
number of directed cycles,
each with at least $d$ vertices.
\par
\noindent
{\em Output:} The shortest-path distances for all
ordered pairs of vertices in $V.$
\noindent
\begin{enumerate}
\item
  Initialize an $n\times n$ array $D$  by setting
  all its entries outside the main diagonal
  to $+\infty$ and those on the diagonal
  to zero.
\item Uniformly at random pick a sample $S$
  of $O(n\ln n/d)$ vertices from $V.$
\item Run the hypothetical APSP algorithm for DAGs on the graph\\
  $(V\setminus S,\ E\cap \{(u,v)|u,v \in V\setminus S\})$
and for each pair $u,\ v\in V\setminus S,$
  set $D(u,v)$ to the distance determined
  by the algorithm.
\item
  For each $s\in S$, run the Dijkstra's SSSP algorithm
  with $s$ as the source in $(V,E)$
  and for all $v\in V\setminus \{s\}$
  update the $D(s,v)$ entries respectively.
\item
 For each $s\in S$, run the Dijkstra's SSSP algorithm
 with $s$ as the source on the directed graph  resulting
 from reversing the directions of the edges in $(V,E),$
  and for all $ v\in V\setminus \{s\}$
  update the $D(v,s)$ entries respectively.
\item
  For all pairs $u,\ v$ of distinct vertices in $V\setminus S,$
  and for all vertices $s\in S,$
  set $D(u,v)=\min \{D(u,v), D(u,s)+D(s,v)\}$.
  \end{enumerate}

\begin{figure}
\label{fig: dag2}
\begin{center}
\includegraphics[scale=0.5]{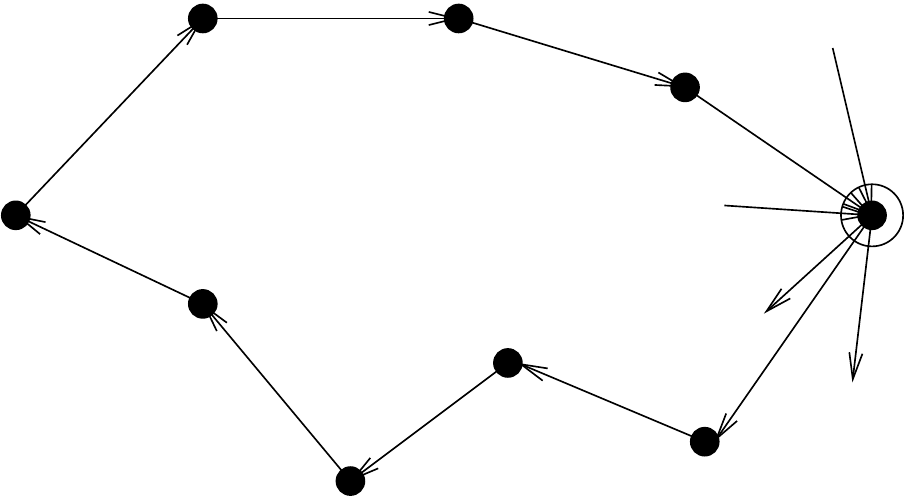}
\end{center}
\caption{An example of a directed cycle that can be broken
  by removing the encircled vertex belonging to the sample.
  To find shortest-path connections passing through this vertex
  two SSSP from it are performed, in the original and the reversed
edge directions, respectively.}
\end{figure}

  \begin{theorem}\label{theo: 3}
Let $t(n,m)$ be the time required by APSP in DAGs with $n$
vertices and $m$ non-negatively weighted
edges.  Algorithm 3 solves the APSP problem for a
    directed graph with $n$ vertices, $m$ non-negatively weighted edges
    and a polynomial number of directed cycles, each with at least $d$
    vertices, in $O(t(n,m)
    +n^3\ln n/d)$ time with high probability.
  \end{theorem}
   \begin{proof}
     Suppose that the number of directed cycles in
     the input graph $(V,E)$ is
$O(n^c).$ By picking enough large constant
for the expression $n\ln n /d$ specifying the size of the sample $S,$
the probability that a given directed cycle in $G$
is not hit by $S$ can be made smaller than $n^{-c-1}.$
Hence, the probability that the graph resulting from
removing the vertices in $S$ is not
acyclic becomes smaller than $n^{-1}.$
It follows that Algorithm 3 is correct with high probability.
It remains to estimate its running time.
Steps 1, 2 can be easily implemented in $O(n^2)$ time.
  Step 3 takes $t(n,m)$ time.
  Steps 4, 5 can be implemented
  in $O((n\ln n/d) \times m +n^2 \ln^2 n /d)$ time \cite{CLR}.
  Finally, Step 6 takes $O(n^3\ln n/d)$ time.
   \end{proof}

   Note that because of the term $n^3\ln n/d$ in the upper time-bound
   given by Theorem \ref{theo: 3}, the upper bound can be substantially
   subcubic only when $d=\Omega (n^{\delta})$ for some $\delta > 0.$
   
\section{Experimental results}

We have implemented Algorithm 1 and the Bellman-Ford algorithm in order to
compare the quality of their estimation of the shortest-path distances
after corresponding iterations. We have also implemented Algorithm 2
and the standard APSP algorithm for DAGS ($n-1$ runs of of the SSSP
dynamic programming algorithm for DAGs) in order to compare their
running times.

For the comparison sake, we used Erdős–Rényi $G(n, p)$ random graph
model, and generated $100$ pseudorandom graphs for $n \in \{10, 100,
1000\}$ and $p \in \{0.2, 0.4, 0.6, 0.8\}$.
We used
\texttt{mt19937} implementation of Mersenne Twister pseudorandom
number generator from GNU C++ Standard Library version 10.2.
Pseudorandom integer
weights from the interval $[-1000, 1000]$ were assigned to the edges.
In case of the APSP algorithms for DAGS, the generated pseudorandom
graphs were converted into DAGs simply by directing each edge
$\{v_i,v_j\}$, where $i<j,$ from $v_i$ to $v_j.$

All four algorithms
were implemented in C++ and Google Benchmark library
was used to measure the CPU time.  High-resolution clock with
nanosecond precision was
used for time measurement.  The code
was
compiled with \texttt{-O2} optimization flag using GNU C++ Compiler
version 10.2, and was executed on a PC with Intel Core i5-2557M 2.7
GHz CPU and 4 GB RAM running Linux kernel version 5.11.15.

\subsection{Algorithm 1}

We have compared the quality of estimations of shortest-path distances
in initial iterations of Algorithm 1 and the Bellman-Ford
algorithm. We count Step 4 as the first iteration, and then each
performance of Steps 5.a and 5.b as consecutive iterations of
Algorithm 1.  In an iteration of the Bellman-Ford algorithm, for each
edge $e$, the current distance (from the source) at the head of $e$ is
compared to the sum of the current distance at the tail of $e$ and the
weight if $e.$ If the sum is smaller the distance at the head of $e$
is updated.

Figures 4 and 5 (see Appendix) show the proportions between the numbers
of vertices for which Algorithm 1 or the Bellman-Ford algorithm
respectively provides a sharper estimation of the shortest-path
distance from the source in corresponding iterations for pseudorandom
graphs on $10$ and $100$ vertices. The figures support the claim
that Algorithm 1 provides reasonable estimation substantially
faster than the Bellman-Ford algorithm does.

\subsection{Algorithm 2}

In one of the initial steps of Algorithm 2, the transitive closure of
the input DAG is computed.  For dense DAGs, the computation of the
transitive closure involves fast matrix multiplication algorithm
known to have huge overhead. Since we run Algorithm 2 on relatively small DAGs
where the aforementioned overhead could shadow the time performance of
the core of the algorithm, we do not account the time taken by the
transitive closure step in our results. See Figure 3.
\junk{
The mean and standard deviation values of the measured CPU time are
reported in the table in Fig. 5.
\begin{figure}[h]
\begin{tabular}{llllll}
\toprule
& & $n=10$ & $n=100$ & $n=1000$ & $n=10000$ \\
\midrule
\multirow{2}{*}{$p=0.2$}
& Baseline & $0.000371\pm0.000035$ & $0.285\pm0.012$ & $236\pm2$ & $238347\pm376$ \\
& Algorithm 2 & $0.000221\pm0.000040$ & $0.162\pm0.009$ & $94\pm2$ & $120679\pm929$ \\
\midrule
\multirow{2}{*}{$p=0.4$}
& Baseline & $0.000504\pm0.000047$ & $0.433\pm0.016$ & $461\pm4$ & $458933\pm1631$ \\
& Algorithm 2 & $0.000353\pm0.000063$ & $0.230\pm0.015$ & $150\pm4$ & $175463\pm947$ \\
\midrule
\multirow{2}{*}{$p=0.6$}
& Baseline & $0.000646\pm0.000053$ & $0.557\pm0.018$ & $697\pm7$ & $683306\pm917$ \\
& Algorithm 2 & $0.000496\pm0.000076$ & $0.277\pm0.0230$ & $196\pm7$ & $221904\pm2371$ \\
\midrule
\multirow{2}{*}{$p=0.8$}
& Baseline & $0.000789\pm0.000052$ & $0.705\pm0.018$ & $945\pm7$ & $913757\pm1281$ \\
& Algorithm 2 & $0.000607\pm0.000089$ & $0.309\pm0.029$ & $233\pm7$ & $259436\pm2522$ \\
\bottomrule
\end{tabular}
\caption{Mean and standard deviation for CPU time in milliseconds}
\end{figure}
}
\begin{figure}[h]
\begin{tabular}{lllll}
\toprule
& & $n=10$ & $n=100$ & $n=1000$ \\
\midrule
\multirow{2}{*}{$p=0.2$}
& Baseline & $0.000371\pm0.000035$ & $0.285\pm0.012$ & $236\pm2$ \\
& Algorithm 2 & $0.000221\pm0.000040$ & $0.162\pm0.009$ & $94\pm2$ \\
\midrule
\multirow{2}{*}{$p=0.4$}
& Baseline & $0.000504\pm0.000047$ & $0.433\pm0.016$ & $461\pm4$ \\
& Algorithm 2 & $0.000353\pm0.000063$ & $0.230\pm0.015$ & $150\pm4$ \\
\midrule
\multirow{2}{*}{$p=0.6$}
& Baseline & $0.000646\pm0.000053$ & $0.557\pm0.018$ & $697\pm7$ \\
& Algorithm 2 & $0.000496\pm0.000076$ & $0.277\pm0.0230$ & $196\pm7$ \\
\midrule
\multirow{2}{*}{$p=0.8$}
& Baseline & $0.000789\pm0.000052$ & $0.705\pm0.018$ & $945\pm7$ \\
& Algorithm 2 & $0.000607\pm0.000089$ & $0.309\pm0.029$ & $233\pm7$\\
\bottomrule
\end{tabular}
\caption{Mean and standard deviation for CPU time in milliseconds}
\end{figure}
\junk{
For further illustration the benchmark results see Fig. 6, 7
in Appendix. Fig. 6 presents a plot showing the mean
and the standard deviation for CPU time versus $n$
in the logarithmic scale. Fig. 7 presents 
scatter plots showing CPU time versus
$m$ (number of edges) for different values of $n.$}
   \section{Final remarks}

   In the absence of substantial asymptotic improvements to the time
   complexity of basic shortest-path algorithms, often formulated at
   the end of 50s, like the Bellman-Ford algorithm and Dijkstra's
   algorithm, the results presented in this paper should be of
   interest. Our output-sensitive algorithm for the general APSP
   problem in DAGs possibly could lead to an improvement of the
   asymptotic time complexity of this problem in the average case.
   A probabilistic analysis of the number of leaves in the
   lexicographically-first shortest-path  trees is an interesting open
   problem.

   In the vast literature on shortest path problems, there are several
   examples of output-sensitive algorithms.  For instance, Karger et
   al. \cite{KKP} and McGeoch \cite{Mc} could orchestrate the $n$ runs of
   Dijkstra's algorithm in order to solve the APSP problem for
   directed graphs with non-negative edge weights in $O(m^*n+n\log n)$
   time, where $m^*$ is the number of (essential) edges that
   participate in shortest paths.

   Finally, note that DAGs
   have several important scientific and computational applications in
   among other things scheduling, data processing networks,
    biology (phylogenetic networks, epidemiology),
    sociology (citation networks), and data compression.
    For these reasons, efficient algorithms for shortest paths in
    DAGs are of not only theoretical interest.


\vfill
\newpage
\section*{Appendix}
\begin{figure}[h]
\label{fig: figesa1}
\begin{center}
\includegraphics[scale=0.5]{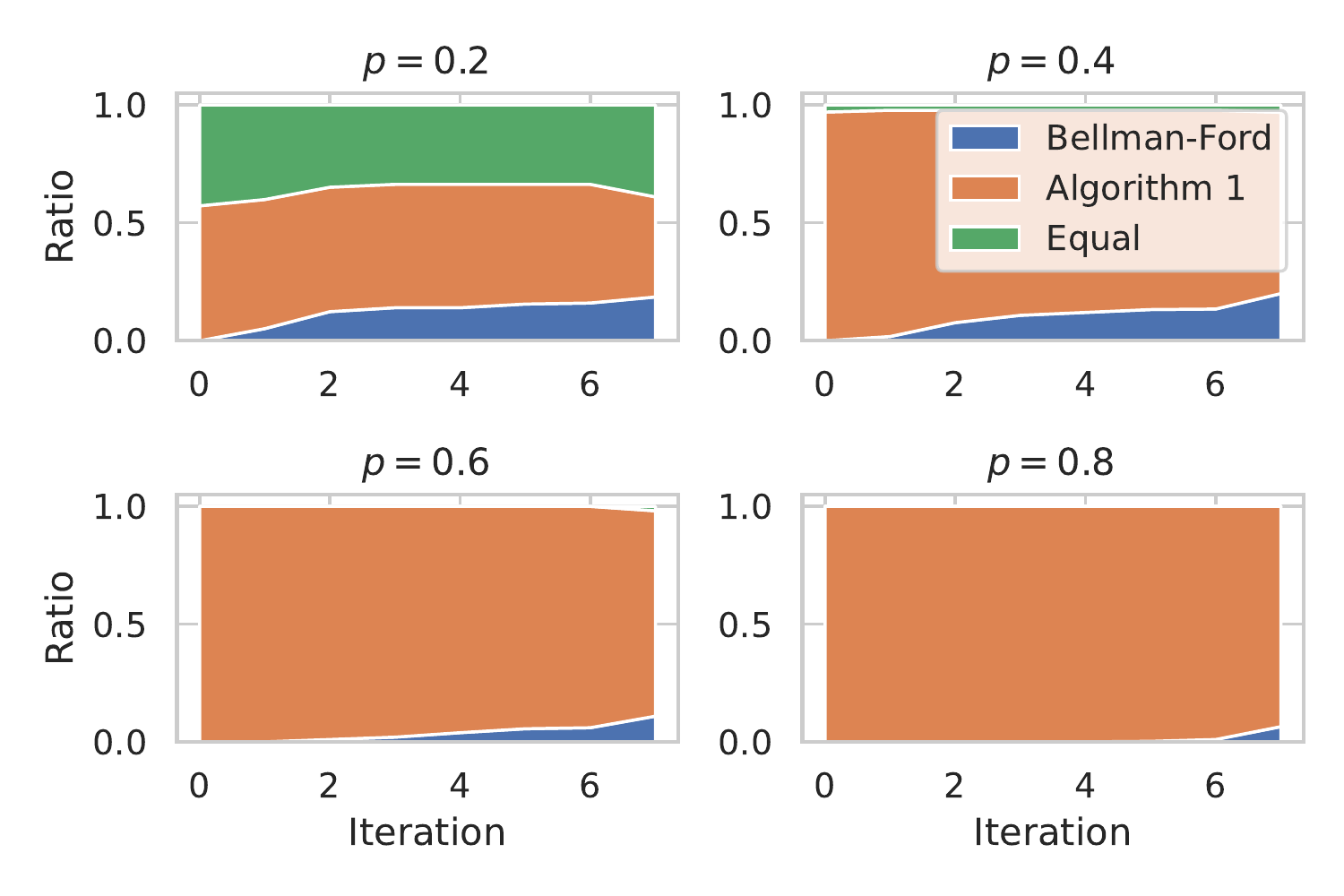}
\end{center}
\caption{A comparison of Algorithm 1 with the Bellman-Ford
  algorithm on pseudorandom graphs with 10 vertices.
  The proportions between the numbers of vertices for
  which Algorithm 1 or the Bellman-Ford algorithm
  respectively gives better estimation of
  the shortest-path distance from the source in corresponding
iterations are visualized with the colors.}
\end{figure}

\begin{figure}[h]
\label{fig: figesa1}
\begin{center}
\includegraphics[scale=0.5]{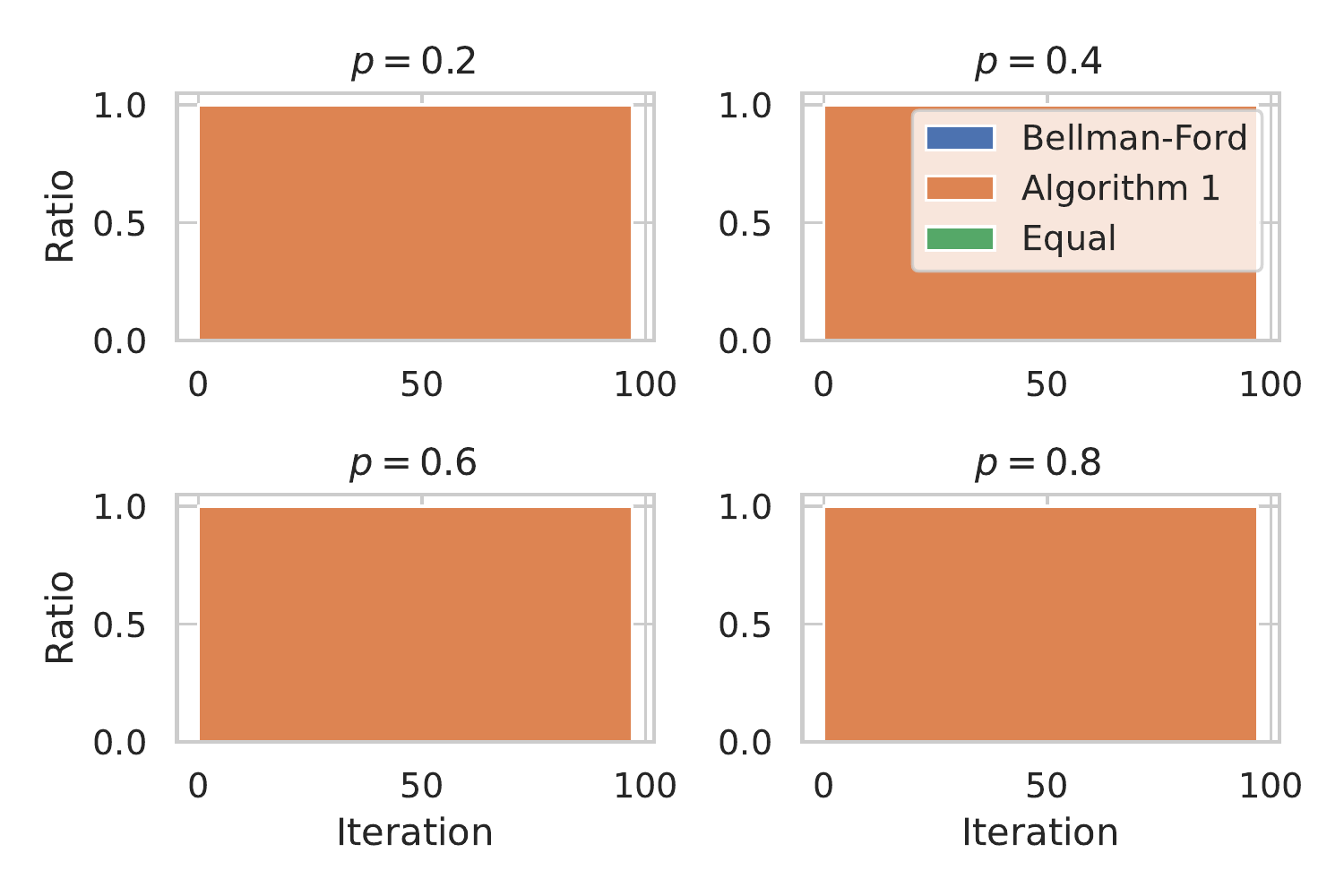}
\end{center}
\caption{An analogous comparison of Algorithm 1 with the Bellman-Ford
  algorithm on pseudorandom graphs with 100 vertices.}
\end{figure}
\end{document}